\documentclass[aip,jmp,superscriptaddress,12pt]{revtex4-1}

\usepackage{amsthm}
\usepackage{amstext}

\usepackage[english]{babel}
\usepackage{mathrsfs}
\usepackage{bm}
\usepackage{amssymb,amsmath}
\usepackage{graphicx}
\usepackage{dcolumn}
\usepackage{color}

\newcommand{\RM}{\mathbb{R}}
\newcommand{\ZM}{\mathbb{Z}}
\newcommand{\QM}{\mathbb{Q}}
\newcommand{\NM}{\mathbb{N}}
\newcommand{\CM}{\mathbb{C}}

\def\e{\mathrm{e}}
\def\i{\mathrm{i}}
\def\d{\mathrm{d}}

\newtheorem{theorem}{Theorem}

\newtheorem{remark}{Remark}

\begin{document}


\title{Large-time limit of the quantum Zeno effect} 

\author{Paolo Facchi}
\email[]{paolo.facchi@ba.infn.it}
\affiliation{Dipartimento di Fisica and MECENAS, Universit\`a di Bari, \\
I-70125 Bari, Italy}

\author{Marilena Ligab\`o}
\email[]{marilena.ligabo@uniba.it}
\affiliation{Dipartimento di Matematica, Universit\`a di Bari, \\
I-70125 Bari, Italy}

\date{\today}

\begin{abstract}
If very frequent periodic measurements ascertain whether a quantum system is still in its initial state, its evolution is hindered.
This peculiar phenomenon is called quantum Zeno effect.
We investigate the large-time limit of the survival probability 
as the total observation time scales as a power of the measurement frequency, $t\propto N^\alpha$.
The limit survival probability exhibits a sudden jump from 1 to 0 at $\alpha=1/2$, the threshold 
between the quantum Zeno effect and a diffusive behavior. 
Moreover, we show that for $\alpha\geq 1$ the limit probability becomes sensitive to the spectral properties of the initial state and to arithmetic properties of the measurement periods.
\end{abstract}

\pacs{03.65.Xp  }
\maketitle 

\section{Introduction}
The evolution of a quantum system is halted when very many measurements are performed in a finite time, in order to check whether the system is still in its initial state. This phenomenon is called Quantum Zeno Effect~\cite{misra} (QZE): the survival probability at a given time goes to one as the measurement frequency increases. 

The survival probability after $N$ measurements in a time $t$ is expressed by a product formula depending on $t$ and $N$. In this paper we  investigate the uniformity in time of the QZE, and study the behavior of the Zeno product formula for large $N$ \emph{and} $t$ .

Let us first recall the basics of the QZE. Let a quantum system be prepared, at time $t = 0$, in the state $\psi$, a normalized vector in the separable Hilbert space $\mathcal{H}$. We denote by $\langle \cdot | \cdot \rangle$ the scalar product in $\mathcal{H}$. The system evolves under the action of the Hamiltonian $H$, a self-adjoint operator on $\mathcal{H}$, through the unitary group $t\mapsto \exp(-\i tH/\hbar)$. The quantities
\begin{equation}
\label{eq:Atdef}
\mathcal{A}(t)=  \langle \psi | \exp\Big(-\frac{\i t}{\hbar}H\Big)\psi \rangle
\end{equation}
and
\begin{equation}
\label{eq:ptdef}
p(t) = \left|\mathcal{A}\left(t\right)\right|^2 = \left|\langle \psi | \exp\Big(-\frac{\i t}{\hbar}H\Big)\psi \rangle \right|^2
\end{equation}
are called survival (or return) amplitude and probability, respectively, and represent the amplitude and probability that the quantum system is found back in the initial state $\psi$ at time $t$. 

If the state~$\psi$ is in the domain of the Hamiltonian~$H$, we have for $t\to0$
\begin{equation}\label{Taylorp}
p(t)= 1-\frac{t^2}{\hbar^2}\left(\langle H \psi | H\psi \rangle - \langle \psi |H\psi \rangle^{2} \right) + o\left( t^2 \right),
\end{equation}
where 
$\langle H \psi | H \psi \rangle - \langle \psi |H\psi \rangle^{2}$ is the variance of the Hamiltonian in the state $\psi$.

Let us now carry out $N$ measurements at time intervals $\tau=t/N$, in order to check whether the system remains in its initial state. If at each and every time the measurement has a positive outcome and the system is found in its initial state, the state ``collapses'' and the evolution starts anew from $\psi$. Thus the
survival probability after  $N$ measurements reads
\begin{equation}
p^{(N)}(t) := p\left(\frac{t}{N }\right)^N= \left|\langle \psi | \exp\Big(-\frac{\i t}{\hbar N}H\Big)\psi \rangle\right|^{2N}.
\label{eq:PNtdef}
\end{equation}
This is called \emph{Zeno product formula}, and will be the subject of our investigation.

The limit of infinitely frequent measurements, $N \to +\infty$, of the Zeno product formula 
can be easily computed using  the Taylor expansion in~(\ref{Taylorp}): if  the initial state $\psi$ is in the domain of the Hamiltonian $H$ one gets
\begin{equation}
\lim_{N\to + \infty} p^{(N)}(t) =1,
\label{eq:QZElimit}
\end{equation}
uniformly in $t$ on compact subsets
of $\RM$, see~\cite {ZenoMP,exner,artzeno}. 
Therefore, if one performs frequent measurements on a quantum system in a given time interval $[0, t]$, a QZE takes place~\cite{misra}:
the transitions to states different from the initial one are hindered,
despite the action of the Hamiltonian (in general the state $\psi$ is not  an eigenstate of the Hamiltonian $H$).

The QZE has been successfully demonstrated in a variety of physical systems, on experiments involving ionic hyperfine levels~\cite{Itano90}, photons~\cite{kwiat}, nuclear spins~\cite{Chapovsky}, optical pumping~\cite{molhave2000},  ultracold atoms~\cite{raizenlatest}, level dynamics of individual ions~\cite{balzer2002},  Bose-Einstein condensates~\cite{ketterle}, optical systems~\cite{hosten}, and cavity quantum electrodynamics~\cite{haroche}. 
For a review on the mathematical and physical aspects of the subject see~\cite {ZenoMP}.

The QZE can be obtained both by pulsed and continuous measurements, as well as by a strong interaction~\cite{Schulman97,Schulman98,zenopraga}.
Recently it has been realized that, by exploiting the quantum Zeno dynamics, one gets a powerful approach to control. The key idea is to engineer a given evolution by a rapid sequence of projections~\cite{AA88,AP89,Bal2000}.
This can yield a Berry phase~\cite{berry} or, more generally, non-Abelian geometric phases~\cite{ShapereWilczek}, a resource for holonomic quantum computation. Moreover, the QZE can be seen as an effective way of imposing constraints and boundary conditions~\cite{regularize,exner,EINZ}.
Finally, notice that the QZE is a purely quantum phenomenon: in classical mechanics it is not observed, since the measurement process can be conceived so that it does not interfere with the evolution of the system.

In this article we want to investigate  the behavior of the Zeno product formula~(\ref{eq:PNtdef}) as the observation time becomes large,
$t \to +\infty$, namely the  double limit: 
\begin{equation}\label{s.c.limit}
\mathop{\lim_ {t \to +\infty  }}_{N \to +\infty}  p^{(N)}(t) .
\end{equation}

Notice that, since the time dependence in~(\ref{eq:PNtdef}) is  given through  the ratio $t/\hbar$, the long-time limit~(\ref{s.c.limit})  is in fact a semiclassical limit, where the Planck constant
$\hbar \to 0$, namely, 
\begin{equation}\label{s.c.limit1}
\mathop{\lim_ {t \to +\infty  }}_{N \to +\infty}  p^{(N)}(t) =  \mathop{\lim_ {t \to +\infty  }}_{N \to +\infty}  \left|\langle \psi | \exp\Big(-\frac{\i t}{\hbar N}H\Big)\psi \rangle\right|^{2N} = \mathop{\lim_ {\hbar \to 0  }}_{N \to +\infty}  \left|\langle \psi | \exp\Big(-\frac{\i t}{\hbar N}H\Big)\psi \rangle\right|^{2N} .
\end{equation}
In this respect, the limit~(\ref{s.c.limit1}) answers  the following
question: what happens to the evolution of the system when we
compare the period between two successive measurements with the
quantum scale given by $\hbar/E_0$, with $E_0$ being the relevant energy scale of the state?

For the analysis of the classical limit of the QZE see~\cite{zenosemiclassic}. 
By a semiclassical analysis on phase space~\cite{tomolectures,L2016}, it can be shown that the QZE vanishes at all orders  in the Planck constant $\hbar$, in the limit $\hbar\to0$,  and thus it is a purely quantum phenomenon without classical analogue, at the same level of tunneling.
(Notice, however, that at variance with~\cite{zenosemiclassic}, in the present situation the state $\psi$ and the Hamiltonian $H$ do not depend on $\hbar$.)

Heuristically, if we perform first the limit in $N$ and
then the limit in $t$ we get QZE, namely
\begin{equation}\label{quantum}
 \lim_{t \to +\infty}\lim_{N \to +\infty}  p^{(N)}(t) =1.
\end{equation}
Conversely, for a decaying system, if we
invert the order of the two limits we obtain a classical behavior, namely 
\begin{equation}\label{classical}
 \lim_{N \to +\infty} \lim_{t \to +\infty} p^{(N)}(t) =0.
\end{equation}
Therefore, the limit~(\ref{s.c.limit}) does not exist because it
depends on the way in which it is computed. 

In order to better understand the transition from~(\ref{quantum}) to~(\ref{classical}), we  look at the double limit~(\ref{s.c.limit}) when $t$ diverges as a power of $N$, i.e. 
\begin{equation}
t=\tau N^{\alpha},
\end{equation} 
where $\tau>0$ is a fixed time interval and $\alpha\geq0$, see the right panel of Fig.~\ref{fig}. In this case the survival probability $p^{(N)}(t)$ depends only on $N$, $\alpha$ and $\tau$. Therefore, in the following we will consider the product formula
\begin{equation}
 p_{N,\alpha}(\tau):= p^{(N)}(\tau N^{\alpha} ) 
 = p(\tau N^{\alpha-1} )^N 
 = \left|\langle \psi | \exp\left(-\frac{\i \tau }{\hbar} N^{\alpha-1} H\right)\psi \rangle \right|^{2N},
 \label{survival prob}
\end{equation}
and  investigate the limit
\begin{equation}\label{eqn:Nlimit}
p^{(\infty)}_{\alpha}(\tau) = \lim_{N \to +\infty} p_{N,\alpha}(\tau),
\end{equation}
for different values of $\alpha \geq 0$.

\begin{figure}
\centering
\resizebox{0.85\textwidth}{!}{\includegraphics{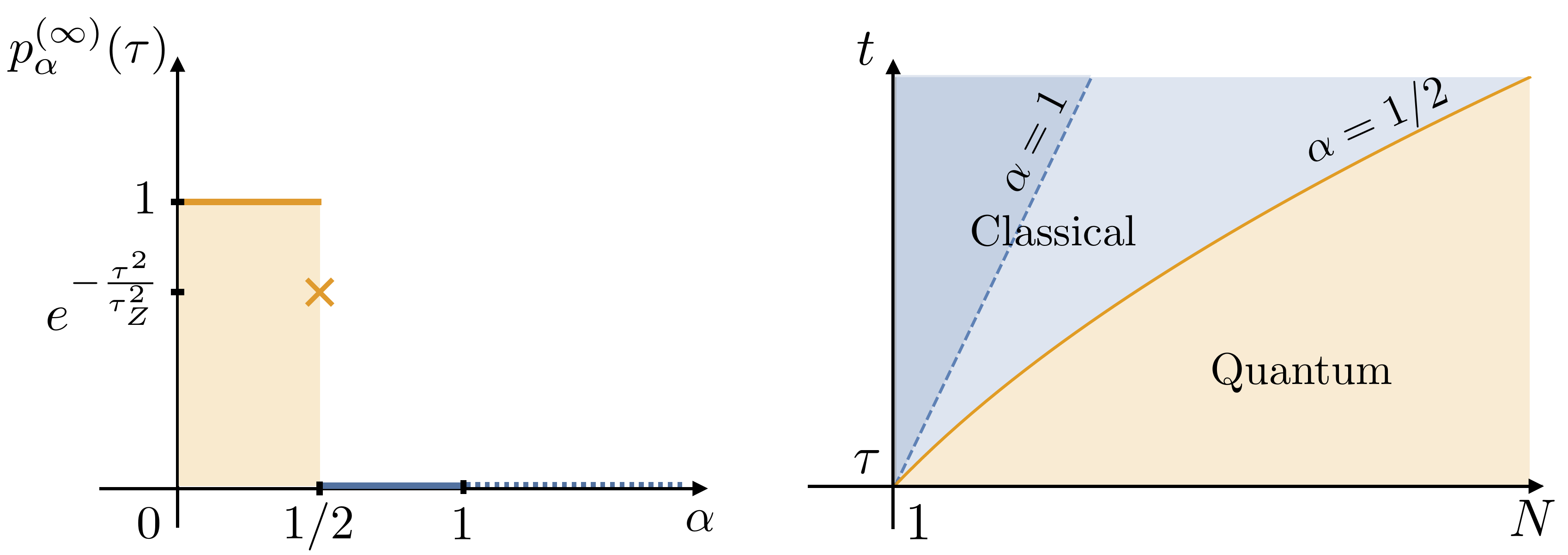}
}
\caption{(Color online) Left panel: The value of the limit $p^{(\infty)}(\tau)$ in~(\ref{eqn:Nlimit}) as a function of the exponent~$\alpha$. Observe the abrupt transition at $\alpha=1/2$. Right panel: quantum and classical regimes in the $N-t$ plane. The QZE effect is along the horizontal axis, but it keeps holding along all curves under the critical parabola $t=\tau N^{1/2}$.}
\label{fig} 
\end{figure}

The value $\alpha=0$ will correspond to the QZE limit~(\ref{eq:QZElimit}), 
\begin{equation}
p^{(\infty)}_{0}(\tau) = \lim_{N\to\infty} p\left(\frac{\tau}{N}\right)^N,
\end{equation}
while the value $\alpha=1$ will correspond to the large-time limit of an evolution stroboscopically measured with period $\tau$,
\begin{equation}
\label{eq:kickslimit}
p^{(\infty)}_{1}(\tau) = \lim_{N\to\infty} p(\tau)^N.
\end{equation} 
The latter regime, describing a quantum system subject to periodic kicks, has become a paradigmatic example in the study of quantum chaos~\cite{Berry79,DegliEspostiGraffi}. It represents a standard test bed for the investigation of different features of quantum systems  whose classical counterparts have a chaotic evolution. In Ref.~\cite{qchaos} the dynamics of a kicked quantum system undergoing repeated measurements of momentum has been investigated. A diffusive behavior has been obtained, even when the dynamics of the classical counterpart is not chaotic, and in general, the system has been shown to have an anomalous diffusive behavior, characteristic of intermittent classical dynamical systems and random walks in random environments~\cite{BCLL2016}.

Thus as $\alpha$ ranges from 0 to 1, one goes from QZE to a kicked dynamics, and for a decaying system the limit probability $p^{(\infty)}_{\alpha}(\tau)$ goes from $1$ to $0$. We will show that the transition is abrupt with a threshold at $\alpha=1/2$, as shown in Fig.~\ref{fig}.

Moreover, we will also consider larger values of the exponent, i.e.\ $\alpha > 1$, which correspond to stroboscopic measurements with a larger and larger period. In such a case  the limit probability $p^{(\infty)}_{\alpha}(\tau)$ of a decaying system is obviously 0, but interesting exceptions will occur at particular values of $\alpha$ and $\tau$ for systems with recurrences.

Notice also that, as a consequence of the previous discussion, the  limit~(\ref{eqn:Nlimit}) can be viewed as the semiclassical limit~(\ref{s.c.limit1}) when $\hbar$ goes to zero as a power of $N$, i.e. 
\begin{equation}
\hbar=\frac{\hbar_{{0}}}{N^{\alpha}},
\end{equation} 
where $\hbar_{{0}}>0$ is a fixed constant. Therefore, for short in the following we will often refer to the regime with zero limit probability as the \emph{classical} regime, as opposed to a nonzero limit probability characteristic of a \emph{quantum} regime, see Fig \ref{fig}.

The article is organized as follows:  in Section \ref{sect:mainresult1} we discuss the case $\alpha < 1$  and we show that $\alpha=1/2$  is the threshold exponent between quantum and classical behavior; in Section \ref{sect:mainresult2} we focus on the case $\alpha \geq 1$ and we prove that, essentially, the system exhibits always a classical behavior, but the limit becomes sensitive to the spectral properties of the state $\psi$ and to some interesting arithmetical properties of $\tau$ and $\alpha$.

\section{Threshold exponent between quantum and classical behavior}\label{sect:mainresult1}
In this section we discuss the case $0 \le \alpha < 1$ and we show that $\alpha=1/2$ is the threshold exponent between quantum and classical behavior.

\begin{theorem}\label{thm:mainth1}
Let $\psi$ be in the domain of $H$, i.e.  $\| H \psi \| < +\infty $. The limit~(\ref{eqn:Nlimit}) of the product formula~(\ref{survival prob}) has the following  behavior:
\begin{enumerate}
    \item[(i)] If $0 \le \alpha < 1/2$ then
\begin{equation}
p^{(\infty)}_{\alpha}(\tau) =1,
\end{equation}
uniformly in $\tau$ on compact subsets of $\RM$;
    \item[(ii)] If
$\alpha=1/2$ 
then
\begin{equation}
p^{(\infty)}_{1/2}(\tau)=\exp\left(-\frac{\tau^2}{ \tau_{Z}^2}\right) ,
\label{eq:Gauss}
\end{equation}
uniformly in $\tau$ on compact subsets of $\RM$, where 
\begin{equation}
\tau_{Z}^{-2}=\frac{1}{\hbar^2}\left( \langle H \psi | H \psi \rangle - \langle \psi |H\psi \rangle^{2}\right) ;
\label{eq:tauZ}
\end{equation}
    \item[(iii)] If
$1/2 < \alpha <1$ and $\psi$ is not an eigenstate of $H$, then
\begin{equation}
p^{(\infty)}_{\alpha}(\tau) =0,
\end{equation}
uniformly in $\tau$ on compact subsets of $\RM \setminus \{0\}$.
\end{enumerate}
\end{theorem}

\begin{remark}
The time $\tau_{Z}>0$ defined by~(\ref{eq:tauZ}) characterizes the initial quadratic behavior of the survival probability and is known in the literature as  the \emph{Zeno time}~\cite{ZenoMP}. Notice that the variance of $H$ given by~(\ref{eq:tauZ}) is zero if and only if the state $\psi$ is an eigenstate of $H$. In such a case the Zeno time is $\tau_Z=+\infty$.
\end{remark}

\begin{proof}
If $0 \le \alpha < 1/2$ then 
\begin{equation}
p_{N,\alpha}(\tau)=p\left(\frac{\tau}{N^{1-\alpha}}\right)^N,
\end{equation}
therefore using~(\ref{Taylorp}) we obtain that
\begin{equation}
p_{N,\alpha}(\tau)  =  \left[ 1-\frac{\tau^2}{ \tau_{Z}^{2} 
N^{2(1-\alpha)} }+o\left( \frac{\tau^2}{ N^{2(1-\alpha)} } \right) 
\right]^N =  1-\frac{\tau^2}{ \tau_{Z}^{2} N^{1-2\alpha}} + o\left( \frac{\tau^2}{ N^{1-2\alpha} } \right). 
\end{equation}
Since $1-2\alpha>0$, we immediately obtain that
\begin{equation}
\lim_{N \to + \infty} p_{N,\alpha}(t)=1,
\end{equation}
uniformly in $\tau$ on compact subsets of $\RM$. 

Following the same procedure we obtain that
for $\alpha=1/2$
\begin{equation}
p_{N,\alpha}(\tau)= \left[ 1-\frac{\tau^2}{\tau_{Z}^{2} N}+o\left(
\frac{\tau^2}{N} \right)\right]^N ,
\end{equation}
therefore
\begin{equation}
\lim_{N \to +\infty} p_{N,\alpha}(\tau) =\exp\left(-\frac{\tau^2}{\tau_{Z}^2}\right)
\end{equation}
uniformly in $\tau$ on compact subsets of $\RM$. 

Finally we discuss the case $1/2 <\alpha <1$. 
Notice that
\begin{equation}
p_{N,\alpha}(\tau) = \exp\left[N \log p(\tau / N^{1-\alpha})\right]
\end{equation}
where $p$ is the survival probability~(\ref{eq:ptdef}).
By~(\ref{Taylorp}), we have
\begin{equation}
\log p(s)=\log\left(1-\frac{s^2}{\tau_Z^2}+o(s^2)\right)=-\frac{s^2}{\tau_Z^2}+o(s^2),
\end{equation}
as $s\to 0$, with a finite $\tau_Z$, since $\psi$ is not an eigenstate of $H$. Therefore for $s$ sufficiently small, say $|s| \leq \sigma$, one gets
\begin{equation}
\log p(s)\leq -\frac{s^2}{2\tau_Z^2},
\end{equation}
whence
\begin{equation}
\log p(\tau / N^{1-\alpha})\leq - \frac{\tau^2}{2\tau_Z^2} \frac{1}{N^{2-2\alpha}},
\end{equation}
for $N\geq (|\tau|/\sigma)^{1/(1-\alpha)}$.

It follows that for $\tau \in [\tau_1,\tau_2] \subset \RM\setminus \{0\}$ one gets
 \begin{equation}
p_{N,\alpha}(\tau) 
\leq \exp\left(- \frac{\tau_{\mathrm{m}}^2}{2\tau_Z^2} N^{2\alpha-1} \right),
\end{equation}
for $N\geq (\tau_{\mathrm{M}}/\sigma)^{1/(1-\alpha)}$, where $\tau_{\mathrm{m}} = \min\{|\tau_1|, |\tau_2|\}>0$ 
and $\tau_{\mathrm{M}} = \max\{|\tau_1|, |\tau_2|\}$.
Therefore,
\begin{equation}
\lim_{N \to +\infty}  p_{N,\alpha}(\tau)= 0
\end{equation}
uniformly in $\tau$ on compact subsets of $\RM \setminus \{0\}$.
\end{proof}

\section{Sensitivity to the spectral  properties of the initial state}\label{sect:mainresult2}
In this section we discuss the case $\alpha \geq 1$. We will show that in this regime the limit~(\ref{eqn:Nlimit}) exhibits always a classical behavior, but it becomes sensitive to the spectral properties of the state $\psi$ and to the arithmetical nature of $\alpha$. 

We recall here some basic aspects of spectral theory; see, e.g.,~\cite{deoliv}. Let $\varphi$ be a vector in the Hilbert space $\mathcal{H}$, and let $H$ be a self-adjoint operator. By the spectral theorem  there exists a unique Borel measure $\mu_{\varphi}$ on $\RM$ such that
\begin{equation}
\langle \varphi | f(H) \varphi \rangle = \int_{\sigma(H)} f(\lambda) \; \d\mu_{\varphi}(\lambda), 
\end{equation}
for all $f \in C_{\mathrm{b}}(\RM)$, where $\sigma(H)$ denotes the spectrum of $H$ and $C_{\mathrm{b}}(\RM)$ denotes the space of bounded and continuous  functions on $\RM$ with complex values. The spectral properties of the Hamiltonian induce a canonical decomposition of the Hilbert space $\mathcal{H}$ into the direct sum
\begin{equation}
\mathcal{H}=\mathcal{H}_{\mathrm{c}} \oplus \mathcal{H}_{\mathrm{pp}}, \qquad \mathcal{H}_{\mathrm{c}}= \mathcal{H}_{\mathrm{ac}}\oplus \mathcal{H}_{\mathrm{sc}},
\end{equation}
where 
\begin{equation*}
\mathcal{H}_{\mathrm{c}}=\{\varphi \in \mathcal{H}: \mu_{\varphi} \text{ is a continuous measure}\} 
\end{equation*}
 is the \emph{continuous subspace},
 \begin{equation*}
\mathcal{H}_{\mathrm{pp}}=\{\varphi \in \mathcal{H}: \mu_{\varphi} \text{ is a pure point measure}\}
\end{equation*}
  is the \emph{pure point subspace},
 \begin{equation*}
\mathcal{H}_{\mathrm{ac}}=\{\varphi \in \mathcal{H}: \mu_{\varphi} \text{ is absolutely continuous}\}
\end{equation*}
 is the \emph{absolutely continuous subspace}, and
\begin{equation*}
\mathcal{H}_{\mathrm{sc}}=\{\varphi \in \mathcal{H}: \mu_{\varphi} \text{ is singular continuous}\},
\end{equation*}
 is the \emph{singular continuous subspace}.
 
 We recall that a Borel measure $\mu_{\mathrm{c}}$ on $\RM$ is  continuous if it does not concentrate at any point, that is if
\begin{equation}
\mu_{\mathrm{c}}(\{x\})=0, \qquad \text{ for all }  x\in\RM,
\end{equation} 
while a measure $\mu_{\mathrm{pp}}$  is pure point (or discrete) if 
\begin{equation}
\mu_{\mathrm{pp}}(B) = \sum_{x\in B} \mu_{\mathrm{pp}}(\{x\}),
\end{equation} 
for all measurable sets $B\subset\mathbb{R}$.
Moreover  a measure $\mu_{\mathrm{ac}}$ is absolutely continuous (with respect to the Lebesgue measure $\d x$)
  if it has a density function $\rho$ locally integrable so that 
\begin{equation}
\d\mu_{\mathrm{ac}}(x) = \rho(x) \d x.
\end{equation}
Finally, a singular continuous measure $\mu_{\mathrm{sc}}$ is continuous but not absolutely continuous; a paradigmatic example is the Cantor measure~\cite{deoliv}.

In this section 
we investigate how the  limit~(\ref{eqn:Nlimit}) changes if the initial state $\psi$ belongs to the spectral subspaces $\mathcal{H}_{\mathrm{pp}}$, $\mathcal{H}_{\mathrm{ac}}$, and $\mathcal{H}_{\mathrm{sc}}$, which physically correspond to bound states (made up of eigenstates), scattering states, and recurring extended states, respectively.

In the following theorem we study the case $\alpha=1$.
\begin{theorem}\label{thm:mainth2}
If $\alpha=1$ the product formula~(\ref{survival prob}) is given by
\begin{equation}
p_{N,1}(\tau)=\left| \langle \psi |\e^{-\frac{\i \tau H}{\hbar }} \psi \rangle \right|^{2N}= p(\tau)^N
\end{equation}
and its limit 
\begin{equation}
p^{(\infty)}_{1}(\tau) = \lim_{N \to +\infty} p(\tau)^N
\end{equation}
has the following  behavior:
\begin{enumerate}
    \item[(i)] If  
    $p(\tau)<1$ for all $\tau \in \RM \setminus\{0\}$
then
\begin{equation}
p^{(\infty)}_{1}(\tau) 
=0,
\end{equation}
uniformly in $\tau$ on compact subsets of $\RM \setminus \{0\}$;
    \item[(ii)] If there exists $\tau_{0}\in \RM \setminus \{0\}$ such
that $p(\tau_0)=1$, 
then $\psi \in \mathcal{\mathcal{H}_{\mathrm{pp}}}$ and there exists a positive integer $m$ such that
\begin{eqnarray}\label{alpha=1 period}
p^{(\infty)}_{1}(\tau)  
= \left\{ \begin{array}{rl} 1 &
\mbox{if} \quad  \tau/\tau_0 \in  \frac{1}{m} \ZM \\
0 & \quad \mbox{otherwise.}
\end{array}
\right.
\end{eqnarray}
\end{enumerate}
\end{theorem}
\begin{proof}
The assertion $(i)$ is obvious. Now  assume that 
\begin{equation}
p(\tau_0)=|\langle \psi| \e^{-\i \tau_{0}H/\hbar }\psi \rangle|^2=1
\end{equation} 
for some 
$\tau_{0}\in \RM \setminus \{0\}$. 
Then  we have that
\begin{equation}
\langle \psi| \e^{-\i \tau_{0}H/\hbar }\psi\rangle=\e^{-\i \tau_{0}a/\hbar }
\end{equation} 
for some $a \in \RM$. By the spectral theorem we have that
\begin{equation}
\int_{\sigma(H)} \e^{-\i \tau_{0}(\lambda-a)/\hbar } \;
\d\mu_{\psi}(\lambda) =1 ,
\end{equation}
and thus
\begin{equation}
\int_{\sigma(H)} (1-\cos(\tau_{0}(\lambda-a)/\hbar )) \;
\d\mu_{\psi}(\lambda)=0.
\end{equation}
Since $(1-\cos(\tau_{0}(\lambda-a)/\hbar )) \geq 0$, we
have that 
\begin{equation}
(1-\cos(\tau_{0}(\lambda-a)/\hbar )) = 0
\end{equation} 
almost everywhere with respect to $\mu_{\psi}$. Therefore, there exist $r \in \NM \cup \{+\infty\}$ and  a subset of integers $\{k_j: j=1,\dots,r\} \subset \ZM$, such that  the support of the spectral measure
$\mu_{\psi}$ is the set
\begin{equation}
\left\{ \lambda_{j}=a+\frac{2 \pi
k_{j}\hbar }{\tau_{0}} : j =1, \dots, r \right\}.
\end{equation}
Therefore the  state $\psi \in \mathcal{H}_{\mathrm{pp}}$ and has the form
\begin{equation}\label{psipp}
 \psi=\sum_{j=1}^{r}c_{j}\psi_{j},
\end{equation}
where $\psi_{j}$ is a normalized eigenvector belonging to  the eigenvalue
$\lambda_{j}$ 
and $c_{j} \in \CM\setminus\{0\}$, for all $j=1, \dots, r$; see~\cite {deoliv}.  Using~(\ref{psipp}) it is easy to check that 
$p(\tau)$ can
be written as follows
\begin{eqnarray}
 p(\tau) & = & \bigg|\int_{\sigma(H)} \e^{-\i \tau \lambda/ \hbar } \;
\d\mu_{\psi}(\lambda)\bigg|^{2} 
          =  \bigg| \sum_{j=1}^rp_j
         \e^{-\i \tau \lambda_{j}/\hbar }
         \bigg|^{2} \nonumber  \\
& = & {\sum_{j, l=1}^r p_{j} p_{l} \cos\left(\frac{2\pi (k_j-k_l)\tau }{\tau_{0}}\right),  } 
\label{eq:survprobpp}
\end{eqnarray}
where $p_{j}=|c_{j}|^2$, $j =1, \dots,r$, with
\begin{equation}
\sum_{j=1}^r p_{j}=1, 
\end{equation}
since $\|\psi\|=1$.
Consider now the greatest common divisor of the $k_j$s
\begin{equation}
m = \gcd\{k_j : j=1,\dots r\},
\end{equation}
so that $k_j= m \tilde{k}_j$ with $\tilde{k}_j\in \ZM$ for $j=1,\dots r$.
We get 
\begin{eqnarray}
 p_{N,1}(\tau) = p(\tau)^N = \bigg[ \sum_{j, l=1}^r p_{j} p_{l} \cos\bigg(\frac{2\pi (\tilde{k}_j-\tilde{k}_l) m \tau}{\tau_{0}}\bigg)   \bigg]^{N} .
\end{eqnarray}
Therefore, if $m \tau/\tau_0$ is an integer then all cosines are equal to one and $p_{N,1}(\tau)=1$ for all $N$. On the other hand,
if $m\tau/\tau_0$ is not an integer then there exists at least a pair of integers $\tilde{k}_j$ and $\tilde{k}_l$ which are coprime, and thus $(\tilde{k}_j-\tilde{k}_l) m \tau/\tau_{0}$ is not an integer and the corresponding cosine is not~$1$. Therefore (\ref{alpha=1 period}) holds. 
\end{proof}
\begin{remark}
Notice that $|\tau_0|/m$ is the first return time of the survival probability~(\ref{eq:ptdef}), that is $p(|\tau_0|/m)=1$ and $p(\tau)<1$ for $0<\tau< |\tau_0|/m$.
Moreover, observe that in the proof of assertion $(ii)$ of the previous Theorem we have retraced the proof of a well known result in probability, see
\cite{shir}, Theorem~5 pag.~288.
\end{remark}
Now we discuss the case $\alpha >1$.
\begin{theorem}\label{mainth3}
If $\alpha>1$ the product formula~(\ref{survival prob}) is given by
\begin{equation}
p_{N,\alpha}(\tau)=\Big| \langle \psi |\exp\Big(-\frac{\i \tau }{\hbar } N^{\alpha-1}H \Big)\psi \rangle \Big|^{2N} \!\!\! = p(\tau N^{\alpha-1})^N
\end{equation}
and its limit~(\ref{eqn:Nlimit}) has the following  behavior:
\begin{enumerate}
\item[(i)] If $H$ is  bounded from below then
\begin{equation}\label{eq:conv0ae}
p^{(\infty)}_{\alpha}(\tau) = \lim_{N \to + \infty} p_{N,\alpha}(\tau)
=0, 
\end{equation}
almost everywhere in $\tau \in \RM$; 
\item[(ii)] If $\psi \in \mathcal{H}_{\mathrm{ac}}$ then the limit~(\ref{eq:conv0ae}) holds
uniformly in $\tau$ on compact subsets of $\RM \setminus \{0\}$;
\item[(iii)] If there exists $\tau_{0}\in \RM \setminus \{0\}$ such
that $p(\tau_0)=1$,  
then $\psi \in \mathcal{\mathcal{H}_{\mathrm{pp}}}$ and
\begin{equation}
\limsup_{N \to + \infty}\,p_{N,\alpha}(M\,\tau_0)=1
\end{equation}for all $\alpha\in\QM$ and all $M \in \ZM$.
\end{enumerate}
\end{theorem}
\begin{proof}
Notice first that if $H$ is bounded from below, say $H\geq E_{\mathrm{min}}$, then $\tilde{H}=H - E_{\mathrm{min}} + 1 \geq 1$ is a strictly positive self-adjoint operator, and
\begin{equation}
| \langle \psi |\e^{-\i t  \tilde{H}/ \hbar}  \psi \rangle |^2 =  | \langle \psi |\e^{-\i t  H/ \hbar}  \psi \rangle |^2 = p(t),
\end{equation}
so we can assume that the Hamiltonian is strictly positive with spectrum 
\begin{equation}
\sigma(H)\subset [1,+\infty).
\end{equation}
Then observe that
\begin{equation}
  \langle \psi | \e^{-\frac{\i \tau N^{\alpha-1}H}{\hbar }} \psi \rangle = \int_{\RM} \exp\left(-\frac{\i \tau N^{\alpha-1}\lambda }{\hbar }\right) \d\mu_{\psi}(\lambda)  
  = \hat{\mu}_{\psi}\left(\frac{N^{\alpha-1}\tau}{\hbar }\right), 
\end{equation}
where $\hat{\mu}$ denotes the Fourier transform of the measure $\mu$. Let $\mu_{\psi,N}$ be  the spectral measure of the self adjoint operator $N^{\alpha-1}H$ in the state $\psi$. Notice that the spectrum of the operator $N^{\alpha-1}H$ is 
\begin{equation}
\sigma(N^{\alpha-1}H)= N^{\alpha-1} \sigma(H) = \{N^{\alpha-1}\lambda : \lambda \in \sigma(H)\}.
\end{equation}
Using the property of the Fourier transform it is easy to see that
\begin{eqnarray}
\left( \langle \psi | \e^{-\frac{\i \tau N^{\alpha-1}H}{\hbar }} \psi \rangle\right)^N = \left[\hat{\mu}_{\psi,N}\left(\frac{\tau}{\hbar }\right)\right]^N =
\hat{\nu}_{\psi,N}\left(\frac{\tau}{\hbar }\right),
\end{eqnarray}
where 
\begin{equation}
\nu_{\psi,N}:= \underbrace{\mu_{\psi,N} \ast \dots \ast  \mu_{\psi,N}}_{N \,\, \mathrm{times}}
\end{equation}
and $\ast$ denotes the convolution product, defined by
\begin{equation}
\int_{\RM} f(\lambda) \; \d\nu_{\psi,N}(\lambda) = \int_{\RM^N}   f(\lambda_1 + \dots + \lambda_N)\, \d\mu_{\psi,N}(\lambda_1) \dots  \d\mu_{\psi,N}(\lambda_N)
\end{equation}
for all $f \in C_{\mathrm{b}}(\RM)$. 

We prove assertion $(i)$.
First we prove that 
\begin{equation}
\lim_{N\to + \infty} \int_{\RM} f(\lambda) \; \d\nu_{\psi,N}(\lambda) =0
\label{eq:limitnu}
\end{equation}
for all $f \in C_0(\RM)$, where $C_{0}(\RM)$ denotes the space of continuous functions vanishing at infinity. 
Indeed, if we fix a function $f \in C_{0}(\RM)$ we have that
\begin{eqnarray}
 \int_{\RM} f(\lambda) \, \d\nu_{\psi,N}(\lambda)
 &  = &  \int_{\RM^N}   f(\lambda_1 + \dots + \lambda_N) \, \d\mu_{\psi,N}(\lambda_1) \dots  \d\mu_{\psi,N}(\lambda_N)  \nonumber \\
&  =&   \int_{\RM^N}  f\bigg(N^{\alpha-1}
\sum_{j=1}^N \lambda_j \bigg)  \d\mu_{\psi}(\lambda_1) \dots \d\mu_{\psi}(\lambda_N)
  \nonumber \\
&  =&  \int_{\RM^N}   f\bigg(N^{\alpha}\sum_{j=1}^N\frac{\lambda_j }{N}\bigg)  \d\mu_{\psi}(\lambda_1) \dots   \d\mu_{\psi}(\lambda_N)  \, \nonumber \\
&  =&  \int_{[1,+\infty)^N}  f\bigg(N^{\alpha}\sum_{j=1}^N\frac{\lambda_j }{N} \bigg)  \d\mu_{\psi}(\lambda_1) \dots  \d\mu_{\psi}(\lambda_N). 
\end{eqnarray} 
Therefore, by the mean value theorem, we have that
\begin{eqnarray}
 \int_{\RM} f(\lambda) \; \d\nu_{\psi,N}(\lambda)   =   f(N^{\alpha} \xi_N),
\end{eqnarray}
where 
\begin{equation}
\xi_N=\frac{\tilde{\lambda}_1+ \dots+ \tilde{\lambda}_N}{N} \ge 1,
\end{equation}
for some $(\tilde{\lambda}_1, \dots, \tilde{\lambda}_N) \in [1,+\infty)^N$ and thus $N^{\alpha} \xi_N \to ~+\infty$ as $N\to+\infty$. 
Therefore, the limit~(\ref{eq:limitnu}) holds.

Now we recall that, by the Riemann-Lebesgue lemma, for all $\phi \in L^1(\RM)$ its Fourier transform $\hat{\phi} \in C_0(\RM)$, and thus
\begin{equation}
\int_{\RM} \phi(\tau) \;\hat{\nu}_{\psi,N}\left(\frac{\tau}{\hbar }\right) \d \tau =\int_{\RM} \hat{\phi}\left(\frac{\lambda}{\hbar }\right) \d\nu_{\psi,N}(\lambda) \to 0 
\end{equation}
as $N \to +\infty$.
Therefore we have proved that 
\begin{eqnarray}
\lim_{N\to +\infty} \hat{\nu}_{\psi,N}\left(\frac{\tau}{\hbar }\right) = \lim_{N\to +\infty} \left( \langle \psi | \e^{-\frac{\i \tau N^{\alpha-1}H}{\hbar }} \psi \rangle\right)^N=0 
\nonumber\\
\end{eqnarray}
almost everywhere in $\tau\in \RM$, whence
\begin{equation}
\lim_{N\to +\infty} p_{N,\alpha}(\tau)= \lim_{N\to+\infty}\left|\hat{\nu}_{\psi,N}\left(\frac{\tau}{\hbar }\right)\right|^2 =0
\end{equation}
almost everywhere in $\tau\in \RM$. 

Now we prove assertion $(ii)$. If we assume that $\psi \in \mathcal{H}_{\mathrm{ac}}$, then $\d\mu_{\psi} (\lambda)= \rho(\lambda)\, \d\lambda$ (where the density $\rho(\lambda) = |\tilde{\psi}(\lambda)|^2 \in L^1(\RM)$ is the squared wave function of $\psi$ in the energy representation), and thus 
\begin{equation}
p(t) =  |\langle \psi, \e^{-\frac{\i t H}{\hbar }} \psi \rangle |^2=
 \left| \int_{\RM} \e^{-\i \frac{t \lambda}{\hbar}} \rho(\lambda) \d\lambda \right|^2 \to 0,
\end{equation}
as $t\to \pm \infty$, by the Riemann-Lebesgue lemma.
Therefore we have that
\begin{equation}
p(\tau N^{\alpha-1}) < \beta <1
\end{equation}
definitively in $N$, for  $\tau \neq 0$, and thus
\begin{equation}
\lim_{N \to +\infty} p_{N,\alpha}(\tau)=\lim_{N \to +\infty} p(\tau N^{\alpha-1})^N = 0
\end{equation}
uniformly in $\tau$ on compact sets of $\RM \setminus \{ 0\}$.

Finally we prove $(iii)$. In the proof of Theorem~\ref{thm:mainth2} we have shown that if the survival probability
\begin{equation}
p(\tau_0)=|\langle \psi| \e^{-\i \tau_{0}H/\hbar }\psi \rangle|^2=1
\end{equation} 
for some $\tau_0 \neq 0$, then $\psi\in\mathcal{H}_{\mathrm{pp}}$ and  the survival probability has the form~(\ref{eq:survprobpp}). Therefore, we have that
\begin{eqnarray}
p_{N,\alpha}(\tau) & = & p(\tau N^{\alpha-1})^N
   =  \bigg| \sum_{j=1}^r p_j
         \e^{-\i \tau N^{\alpha-1}\lambda_{j}/\hbar }
         \bigg|^{2N}  \nonumber \\
         & = & \bigg[\sum_{j, l=1}^r p_{j} p_{l} \cos\bigg(\frac{2\pi (k_j-k_l)N^{\alpha-1}\tau}{\tau_{0}}\bigg)
         \bigg]^{N},
\end{eqnarray}
with $\sum p_{j}=1$ and $k_j \in \ZM$ for $j=1,\dots r$. Notice that if $\alpha \in \QM$, then $\alpha-1=n_1/n_2$ for some $n_1, n_2 \in \NM \setminus \{0\}$. Therefore if we consider the subsequence $N_m:=m^{n_2}$ we have that
\begin{equation}
p_{N_m, \alpha}(M\, \tau_0) 
= \bigg[\sum_{j, l=1}^r p_{j} p_{l} \cos\big(2\pi (k_j-k_l)m^{n_1} M\big)
         \bigg]^{N_m} =1,
\end{equation}
for all $m \in \NM$ and $M \in \ZM$, therefore we have that
\begin{equation}
\lim_{m \to +\infty} p_{N_m, \alpha}(M\, \tau_0)=1.
\end{equation}
Since $p_{N, \alpha}(\tau) \leq 1$ for all $\tau \in \RM$ and $N \in \NM$, we conclude that
\begin{equation}
 \lim_{m \to +\infty} p_{N_m, \alpha}(M\, \tau_0)=\limsup_{N \to +\infty} \,p_{N, \alpha}(M\, \tau_0)=1.
\end{equation}
\end{proof}

\section{Conclusions}

Let us summarize the main results obtained in this article in more intuitive terms, by focusing on those quantities that are more directly related to physical intuition. We have analyzed the double limit
\begin{equation}\label{limitcon}
 \mathop{\lim_{t \to +\infty  }}_{N \to +\infty}  p^{(N)}\left( t \right)
\end{equation}
in the case $t \propto N^{\alpha}$, for all possible values of $\alpha\ge0$. 

We have shown that if $0\le \alpha<1/2$ the limit equals $1$, namely the system is frozen in its initial state and  the QZE takes place. 
At $\alpha=1/2$ the limit~(\ref{limitcon}) is strictly smaller than $1$ for all times  
and decays in time as a Gaussian~(\ref{eq:Gauss}). 
If $1/2<\alpha <1$ the limit~(\ref{limitcon}) is  equal to $0$ for  all times, 
 thus we observe a classical behavior.  See Fig.~\ref{fig}.

Moreover, if $\alpha \geq 1$ the limit probability is a strange beast and  becomes sensitive to the spectral properties of the state $\psi$. In general, the limit~(\ref{limitcon})  is $0$ for almost all times, and  if the state is decaying, $\psi \in \mathcal{H}_{\mathrm{ac}}$, 
the limit~(\ref{limitcon}) is always $0$ for all times.  

The existence of times $t$ at which the limit~(\ref{limitcon}) is nonzero has been clarified, at least for bound states, $\psi\in\mathcal{H}_{\mathrm{pp}}$. In fact there are bound states with a periodic dynamics. Thus, if one
performs repeated measurements at the natural period, namely if one looks stroboscopically at the particle dynamics, 
the presence of the measurements becomes immaterial: the
classical and the quantum behavior simply coincide and the limit~(\ref{limitcon}) at that times is equal to $1$. 

Concerning the existence of times 
at which the limit~(\ref{limitcon}) is not  $0$  for states in the continuous singular spectrum $\psi \in \mathcal{H}_{\mathrm{sc}}$, i.e.\  recurrent unbounded states, we can only say that the set of those times is negligible. It is more difficult to grasp this situation by physical intuition, and its full comprehension would require a further analysis which is beyond the scope of this paper.

Summarizing, we have unveiled the presence of two threshold exponents: the threshold  between  quantum and  classical behavior at $\alpha=1/2$, and the threshold of sensitivity to the spectral properties of the initial state at $\alpha=1$.

\medskip
\begin{acknowledgments}
This work was  supported by Cohesion and Development Fund 2007-2013 - APQ Research Puglia Region ``Regional program supporting smart specialization and social and environmental sustainability - FutureInResearch'', by the Italian National Group of Mathematical Physics (GNFM-INdAM), and by Istituto Nazionale di Fisica Nucleare (INFN) through the project ``QUANTUM''.
\end{acknowledgments}


\begin{thebibliography}{0}
\bibitem{misra} Misra B. and Sudarshan E.C.G.,
{``The Zeno's paradox in quantum theory''}, {J. Math. Phys.}
  \textbf{18}, 756 (1977)

\bibitem{ZenoMP}
Facchi P. and Pascazio S., {``Quantum Zeno  dynamics: mathematical and physical aspects''}, {J. Phys. A: Math. Theor.} \textbf{41}, 493001 (2008)

\bibitem{exner}
Exner P. and Ichinose T.,
``A product formula related to quantum Zeno dynamics'', {Ann. Henri
Poincar\'e}
  \textbf{6}, 195 (2005)

\bibitem{artzeno}
Facchi P. and Ligab\`o  M., ``Quantum Zeno effect and dynamics'', {J. Math. Phys.} \textbf{51}, 022103 (2010)

\bibitem{Itano90}
Itano W.M., Heinzen D.J., Bollinger J.J. and Wineland D.J.,
``Quantum Zeno effect'',
{Phys. Rev. A} \textbf{41}, 2295 (1990)

\bibitem{kwiat}
Kwiat R., Weinfurter H., Herzog T., Zeilinger A., and Kasevich M.,
``Interaction-Free Measurement'',
{Phys. Rev. Lett.} \textbf{74}, 4763 (1995)

\bibitem{Chapovsky}
Nagels B., Hermans L.J.F.  and Chapovsky P.L.,
``Quantum Zeno Effect Induced by Collisions'', 
{Phys. Rev. Lett.} \textbf{79}, 3097 (1997) 

\bibitem{molhave2000}
M{\o}lhave K. and Drewsen M.,
``Demonstration of the continuous quantum Zeno effect in optical pumping'',  
{Phys. Lett.} A {\bf 268}, 45 (2000)

\bibitem{raizenlatest}
Fischer M.C., Guti\'errez-Medina B. and Raizen M.G.,
``Observation of the Quantum Zeno and Anti-Zeno Effects in an Unstable System'',
{Phys. Rev. Lett.} \textbf{87}, 040402 (2001) 

\bibitem{balzer2002}
Balzer C., Hannemann T., Reib D., Wunderlich C., Neuhauser W. and Toschek P.E.,
``A relaxationless demonstration of the Quantum Zeno paradox on an individual atom'',
{Opt. Commun.} \textbf{211}, 235 (2002)

\bibitem{ketterle}
Streed E.W., Mun J., Boyd M., Campbell G.K., Medley P., Ketterle W. and
Pritchard D.E.,
``Continuous and Pulsed Quantum Zeno Effect'',
{Phys. Rev. Lett.} \textbf{97}, 260402 (2006)

\bibitem{hosten}  
Hosten O.,  Rakher M.T., Barreiro J.T., Peters N.A. and Kwiat P.G.,
``Counterfactual quantum computation through quantum interrogation'',
Nature \textbf{439}, 949 (2006).

\bibitem{haroche}
Bernu J., Del\'eglise S., Sayrin C., Kuhr S., Dotsenko I., Brune M., Raimond, J.M. and Haroche S.,
``Freezing Coherent Field Growth in a Cavity by the Quantum Zeno Effect'',
{Phys. Rev. Lett.} \textbf{101}, 180402 (2008)

\bibitem{Schulman97}
Schulman L.S.,
``Observational line broadening and the duration of a quantum jump'',
{J. Phys. A} \textbf{30}, L293 (1997)

\bibitem{Schulman98}
Schulman L.S.,
``Continuous and pulsed observations in the quantum Zeno effect'', 
{Phys. Rev. A} \textbf{57}, 1509 (1998).

\bibitem{zenopraga}
Facchi P., Pascazio S.,
``Quantum Zeno phenomena: Pulsed versus continuous measurement'',
{Fortschr. Phys.} \textbf{49}, 941 (2001) 

\bibitem{AA88}
Anandan J. and Aharonov Y.,
``Geometric quantum phase and angle'',
{Phys. Rev. D} \textbf{38}, 1863 (1988)

\bibitem{AP89}
Anandan J. and Pines A.,
``Non-Abelian geometric phase from incomplete quantum measurements'',
{Phys. Lett. A} \textbf{141}, 335 (1989)

\bibitem{Bal2000} 
Balachandran A.P. and Roy S.M.,
``Quantum Anti-Zeno Paradox'',
{Phys. Rev. Lett.} \textbf{84}, 4019 (2000)

\bibitem{berry}
Facchi P., Klein A.G., Pascazio S., Schulman L.S.,
``Berry phase from a quantum Zeno effect'', 
{Phys. Lett. A}   \textbf{257}, 232  (1999) 

\bibitem{ShapereWilczek} 
Shapere A. and Wilczek F.,
\textit{Geometric Phases in Physics}
 ({World Scientific, Singapore} 1989) 

\bibitem{regularize}
Facchi P., Pascazio S., Scardicchio A., Schulman L.S.,
``Zeno dynamics yields ordinary constraints'',
{Phys. Rev. A}   \textbf{65}, 012108 (2002)  

\bibitem{EINZ}
Exner P., Ichinose T., Neidhardt H. and Zagrebnov V.,
``Zeno Product Formula Revisited'',
{Integr. Equ. Oper. Theory}
\textbf{57}, 67 (2007)

\bibitem{zenosemiclassic}
Facchi P., Graffi S. and Ligab\`o  M., 
``The classical limit of the quantum Zeno effect'', 
{J. Phys. A: Math. Theor.} \textbf{43}, 032001 (2010) 

\bibitem{tomolectures}
Facchi P. and Ligab\`{o} M.,
``Classical and quantum aspects of tomography'',
{AIP Conf. Proc.} \textbf{1260}, 3 (2010)

\bibitem{L2016}
Ligab\`{o} M.,
``Torus as phase space: Weyl quantization, dequantization, and Wigner formalism'',
{J. Math. Phys.} \textbf{57}, 082110 (2016)

\bibitem{Berry79}
Berry M.V., Balazs N.L., Tabor M. and Voros A., 
``Quantum Maps'',
{Ann. Phys. (N.Y.)} \textbf{122}, 26 (1979)

\bibitem{DegliEspostiGraffi}
Degli Esposti M. and  Graffi S.,
\textit{The Mathematical Aspects of Quantum Maps},
Lectures Notes in Physics (Springer, 2003)

\bibitem{qchaos}
Facchi P., Pascazio S. and Scardicchio A.,
``Measurement-induced quantum diffusion'',
{Phys. Rev. Lett.} \textbf{83}, 61 (1999)

\bibitem{BCLL2016}
Bianchi A., Cristadoro G., Lenci M., Ligab\`{o} M.,
``Random Walks in a One-Dimensional L\'{e}vy Random Environment'',
{J. Stat. Phys.} \textbf{163}, 22 (2016)

\bibitem{deoliv}
De Olivera C. R.,  \textit{Intermediate spectral theory and quantum dynamics} (Progress in Mathematical Physics
Vol.\textbf{54}, Birkh\"auser 2009)

\bibitem{shir}
Shiryaev A. N.: \textit{Probability} (Springer, 1995)


\end{thebibliography}

\end{document}